\documentclass[11pt]{amsart}
\usepackage[foot]{amsaddr}   
     
\usepackage{amsmath,amssymb}
\usepackage{xcolor}   
 
\newtheorem{theorem}{Theorem}[section]   
\newtheorem{lemma}[theorem]{Lemma}   
\newtheorem{proposition}[theorem]{Proposition}
\newtheorem{corollary}[theorem]{Corollary}

\theoremstyle{definition}
\newtheorem{definition}[theorem]{Definition} 

\newtheorem{remark}{Remark}[section]

\usepackage[right=1in, left=1in, top=1in, bottom=1in, marginpar=0.75in]{geometry}

\renewcommand{\leq}{\leqslant}  
\renewcommand{\geq}{\geqslant}
\newcommand{\Aut}{{\textrm{\rm Aut}}}
\newcommand{\Cay}{{\textrm{\rm SCG}}}

\newcommand{\ourset}{\Upsilon}

\usepackage{array} 
\newcolumntype{?}{!{\vrule width 1.5pt}}

\newcommand{\mydes}[5]{

  \smallskip
  \hspace*{-1.3ex}\begin{tabular}{@{}c?lp{12.5cm}}
    \multicolumn{3}{l}{{\textsc{#1}}}\\
   \hspace*{1ex} & {\bf #2} & #3\\
    & {\bf #4} & #5\\
  \end{tabular}
  
  \smallskip

  }

\begin{document}

\title{Group isomorphism is nearly-linear time for most orders}
 
\author{Heiko Dietrich}
\email{heiko.dietrich@monash.edu}
\address{%
   School of Mathematics, Monash University, Australia
 }

\author{James B.\ Wilson}
\email{james.wilson@colostate.edu}
\address{Department of Mathematics, Colorado State University, Colorado, USA}
   
\begin{abstract}
  We show that there is a dense set $\ourset\subseteq \mathbb{N}$ of group
  orders and a constant $c$ such that for every $n\in \ourset$ we can decide in
  time $O(n^2(\log n)^c)$ whether two $n\times n$ multiplication tables describe
  isomorphic groups of order $n$.  This improves significantly over the general
  $n^{O(\log n)}$-time complexity and shows that group isomorphism can be tested
  efficiently for almost all group orders $n$. We also show that in time $O(n^2
  (\log n)^c)$ it can be decided whether an $n\times n$ multiplication table
  describes a group; this improves over the known $O(n^3)$ complexity. Our
  complexities are calculated for a deterministic multi-tape Turing machine
  model.  We give the implications to a RAM model in the promise hierarchy as well.
\end{abstract}


\maketitle

\newcommand{\myl}{{{\text{\rm\tiny big}}}}

\newcommand{\ppi}{{\pi^{{\rm\tiny si}}}}
\newcommand{\bppi}{{\pi^{{\rm\tiny lsi}}}}

\section{Introduction}
\label{sec:intro}
  
\noindent Given a natural number $n$, there are many structures that can be recorded by an
$n\times n$ table $T$ taking values $T_{ij}$ in $[n]=\{1,\ldots,n\}$.
Isomorphisms of these tables are permutations $\sigma$ on $[n]$ with
$T_{\sigma(i)\sigma(j)}=\sigma(T_{ij})$ for all $i,j\in[n]$.    It is convenient
to assign these tables either a geometric or algebraic interpretation.  A
geometric view treats these as edge colored directed graphs or as the relations
of an incidence structure.  We will consider  the algebraic interpretation where
the table describes a binary product $*\colon [n]\times [n]\to [n]$.

An upper bound complexity to decide isomorphism can be given by testing all $n!$
permutations.  Better timings arise when we consider subfamilies of structures,
for example, by imposing equational laws on the product such as associativity
$a*(b*c)=(a*b)*c$ (i.e.\ \emph{semigroups}) or the existence of left and right
fractions (i.e.\ \emph{quasigroups} or \emph{latin squares}). Booth
\cite[p.~132]{miller79}\linebreak observed that the complexity of isomorphism
testing of semigroups is polynomial-time equivalent to the complexity of graph
isomorphism. At the time, that complexity was subexponential, but it has since
been shown by Babai \cite{Babai:quasi} to be in quasi-polynomial time with  a
highly inventive algorithm.  Meanwhile Miller \cite{miller78,miller79} observed
that the complexity of quasigroup isomorphism is in quasi-polynomial time
$n^{O(\log n)}$ through an almost brute-force algorithm: since quasigroups of
order $n$ are generated by $\log_2 n$ elements, a brute-force comparison of all
$\lceil \log_2 n\rceil$-tuples either finds an isomorphism between two
quasigroups or determines that no isomorphism exists.

\enlargethispage{3ex}

An intriguing bottleneck to further improvement has been the case of
\emph{groups} that have associative products with an identity and left and right
inverses. Because these are quasigroups, they have a brute-force isomorphism
test with complexity $n^{O(\log n)}$;  Miller credited Tarjan for the complexity 
of group isomorphism.
Guralnick and Lucchini \cite[Theorem 16.6]{BNV:enum} showed
independently  that every finite group of order $n$ can be generated by at most
$d(n)$ elements and $d(n) \leq \mu(n)+1$, where $\mu(n)$ is the largest exponent
of a prime power divisor $p^{\mu(n)}$ of $n$. Thus, the complexity of
brute-force group isomorphism testing is more accurately described as
$n^{O(\mu(n))}$. Since  $\mu(n)\in \Theta(\log n)$ when $n=2^{\ell}m$ with $m\in O(1)$,
this does not improve on the generic $n^{O(\log n)}$ bound.  Group isomorphism
testing seems to be a leading bottleneck to improving the complexity of graph
isomorphism, see Babai \cite[Section 13.2]{Babai:quasi}. Even so, we prove here
that  for \emph{most} orders, group isomorphism is in nearly-linear time
compared to the input size; this also shows that groups of general orders are not a
bottleneck to graph isomorphism.

\subsection{Current state of  group isomorphism}
Surprisingly, the brute-force $n^{O(\log n)}$ complexity of group isomorphism
has been resilient.  Progress has fragmented into work on numerous
subclasses~$\mathfrak{X}$ of groups;  the precise problem studied today
is:

\medskip

\mydes{$\mathfrak{X}$-GroupIso}
  {Given}{
    a pair $(T,T')$ of $n\times n$ tables with entries in 
    $[n]$ representing groups in $\mathfrak{X}$,
  }
  {Decide}{if the groups are isomorphic.}

  \medskip

\noindent We follow the convention that an $n\times n$ table with
entries in $[n]$ corresponds to the multiplication table of a binary product
where the rows and columns are both labelled by $1,2,\ldots,n$. Moreover, for
a group multiplication we require that the identity element is denoted
``$1$'', that is, the first row and first column must both be $1,2,\ldots,n$;
if this is not the case, then we reject the input.

Grochow-Qiao \cite{GQ} give a detailed survey of recent progress on group
isomorphism; here we summarize a few results related to our setting. Iliopoulos
\cite{Ili}, Karagiorgos-Poulakis \cite{KP}, Vikas \cite{vikas}, and Kavitha
\cite{kavitha} progressively improved the complexity for the class
$\mathfrak{A}$ of abelian groups (where the product satisfies $a*b=b*a$ for all
$a,b$), resulting in a linear-time algorithm for
$\mathfrak{A}$-\textsc{GroupIso} in a RAM model (more on this below).
Wagner-Rosenbaum \cite{rosenbaum} gave an $n^{0.25 \log n+O(1)}$ time algorithm
for the class $\mathfrak{N}_p$ of groups of order a power of a prime $p$, and
later generalized this to the class of solvable groups.  Li-Qiao \cite{Li-Qiao}
proved an average run time of $n^{O(1)}$ for an essentially dense subclass
$\mathfrak{N}_{p,2}\subset \mathfrak{N}_p$. Babai-Codenotti-Qiao~\cite{BCQ}
proved an $n^{O(1)}$ bound for the class $\mathfrak{T}$ of groups with no
nontrivial abelian normal subgroups.  Das-Sharma \cite{das} described a
nearly-linear time algorithm for the class of Hamiltonian groups, again in a RAM
model.

Other research  combines results for separate
classes by considering isomorphism between groups that decompose into a subgroup
in class $\mathfrak{X}$ and a quotient in class $\mathfrak{Y}$,  see also
Section~\ref{sec:prel}; we call this the
\emph{$(\mathfrak{X},\mathfrak{Y})$-extension problem}.  Le Gall
\cite{LeGall} studied $(\mathfrak{A},\mathfrak{C})$-extensions,
where $\mathfrak{C}$ consists of cyclic groups. Grochow-Qiao \cite{GQ}
considered $(\mathfrak{A},\mathfrak{T})$-extensions, and outlined a general
framework for solving  extension problems.

A further class of algorithms considers terse input models, such as black-box
models or groups of matrices or permutations; we refer to Seress
\cite[Section~2]{Seress} for details of those models. In this
format, groups can be exponentially larger than the data it takes to specify the
group.  Using this model, the second author proved an $(\log
n)^{O(1)}$-time algorithm for subgroups and quotients of finite Heisenberg
groups, and further variations in collaborations with Lewis and
Brooksbank-Maglione; see \cite{Wilson:profile} and the references therein.
Recently, in \cite{PART2} the authors proved a polynomial-time isomorphism test
for permutation groups of square-free and cube-free orders. These examples
demonstrate that input models may have an outsized influence on the complexity
of group isomorphism.

Some of the motivation of this and earlier work \cite{PART2} has been the
observation that, in contrast to graph isomorphism, the difficulty of group
isomorphism is influenced by the prime power factorization of the group orders
$n$.  For example, if $n=2^e\pm 1$ is a prime, then there is exactly one
isomorphism type of groups of prime order $n$ and isomorphism can be tested by
comparing orders.  Yet, there are $n^{2e^2/27-O(e)}$ isomorphism types of groups
of prime power order $n\mp 1=2^e$, see \cite[p.~23]{BNV:enum}. As of today,
isomorphism testing of groups of order $2^e$ has the worst-case complexity.
  
\subsection{Main results}
The main result of this paper is a proof that group isomorphism can be tested
efficiently for \emph{almost all} group orders $n$ in time $O(n^2 (\log n)^c)$
for some constant $c$, if the groups are input  by their Cayley tables, that is,
by $n\times n$ tables describing their  multiplication maps $[n]\times [n]\to
[n]$.   To make ``almost all'' specific, we define the \emph{density} of a set
$\Omega\subseteq\mathbb{N}$ to be the limit $\delta(\Omega)=\lim_{n\to
\infty}|\Omega\cap [n]|/n$; the set $\Omega$ is \emph{dense} if
$\delta(\Omega)=1$. By abuse of notation, $\Omega$-\textsc{GroupIso} denotes the
isomorphism problem for the class of groups whose orders lie in $\Omega$.
All our complexities are stated for deterministic multi-tape Turing machines;
we give details in Section~\ref{sec:prel}.

\begin{theorem}\label{thm:main} 
  There is a dense subset  $\Upsilon\subset\mathbb{N}$ and a deterministic
  Turing machine that decides $\Upsilon$-\textsc{GroupIso} for $n\in\Upsilon$ in time 
  $O(n^2(\log n)^c)$ for some constant $c$. 
\end{theorem}

We provide a proof in Section~\ref{appthm:main}; the set $\ourset$ is specified
in Definition \ref{def:ups} and motivated by the Hardy-Ramanujan Theorem
\cite[Section 8]{hardy} and number theory results of Erd\H{o}s-P\'alfy
\cite{ErdosPalfy}. Since every multiplication table $[n]\times [n]\to [n]$ can
be encoded and recognized from a binary string of length $\Theta(n^2\log n)$,
the algorithm of Theorem~\ref{thm:main} is nearly-linear time in the input size.
The dense set $\Upsilon$ is specified in Definition~\ref{def:ups}; here we
mention that we can determine in time $O(n^2 (\log n)^c)$ whether
$n\in\Upsilon$, see Remark \ref{remUps}, and the complexity for brute-force
isomorphism testing of groups of order $n\in \Upsilon$ is $n^{O(\log\log n)}$.
Because of this, we would have been content with a polynomial-time bound; being
able to prove nearly-linear time bound was a surprise. 

 Our set $\Upsilon$ excludes an important but difficult class of group orders,
 specifically orders that have a large power of a prime as a divisor.
 Theorem~\ref{thm:main} therefore  goes some way towards confirming the
 expectation that groups of prime power order are the essential bottleneck to
 group isomorphism testing. Indeed, examples such as provided in
 \cite{Wilson:profile} show that large numbers of groups of prime power order
 can appear identical and yet be pairwise non-isomorphic. In fact, known
 estimates on the proportions of groups show that most isomorphism types of
 groups accumulate around orders with large prime powers, see 
 \cite[pp.\ 1--2]{BNV:enum}. So our Theorem~\ref{thm:main} should not be misunderstood as
 saying that group isomorphism is efficient on most groups, just on most orders.
 Even so, we see in results like Li-Qiao \cite{Li-Qiao} and
 Theorem~\ref{thm:main} the beginnings of an approach to show that group
 isomorphism is polynomial-time on average, and we encourage work in this
 direction.

The solutions of $\mathfrak{X}$-\textsc{GroupIso} cited so far deal with the
problem in the promise polynomial hierarchy \cite{promise} where one promises
that inputs are known to be groups and that they lie in $\mathfrak{X}$.  To
relate those solutions to the usual deterministic polynomial-time hierarchy
forces us to consider the complexity of the associated membership problem:

\medskip 

\mydes{$\mathfrak{X}$-Group}
  {Given}{
    a binary string $T$,
  }
  {Decide}{if $T$ encodes the Cayley table of a group contained in $\mathfrak{X}$.}

\medskip
  
\noindent 
Current methods in the literature seem to require $O(n^3)$ steps to
verify that a binary product on $[n]$ is associative, see \cite[Chapter
2]{latin}.  Here we present a method that solves
$\mathfrak{G}$-\textsc{Group} for the class $\mathfrak{G}$ of all finite groups
in time $O(n^2(\log n)^d)$ for some constant $d$. We note that
Rajagopalan-Schulman \cite[Theorem 5.2]{ids} provide an $O(n^2\log n)$ algorithm
for this task, but they cost the binary operation as $O(1)$, which gives an
upper bound of $O(n^4(\log n)^2)$ on a Turing machine~model.
\enlargethispage{2ex}

\begin{theorem}\label{thm:isgrp} There is a deterministic Turing machine that
  decides in time ${O}(n^2(\log n)^d)$ for some constant $d$ whether an $n\times
  n$ table with entries in $[n]$ describes a group and, if so, returns a
  homomorphism $[n]\to {\rm Sym}_n$ into the group ${\rm Sym}_n$ of permutations
  on $[n]$.
\end{theorem}
Recall our assumption that the multiplication table arising from the input table
has labels $1,2,\ldots, n$ and that ``$1$'' must be the identity, that is, the
first row and column of the table must be $1,2,\ldots,n$; if not, then the input
is rejected as it does not represent a group.   Relaxing that assumption leads
to comparing multiplication tables with independent permutations of the rows,
columns, and entries producing what is known as an isotopy (or isotopism)
instead of an isomorphism.  We consider only isomorphism. 

We prove Theorem~\ref{thm:isgrp} in Appendix~\ref{app:isgroup}.  From that
result, Theorem~\ref{thm:main} can be described as nearly-linear time on all
inputs, that is, it properly accepts or rejects all strings without assuming
external promises on these inputs (such as that the tables represent groups
or groups with some property). Theorem~\ref{thm:isgrp} also offers a hint that
our strategy partly entails working with  data structures for permutation
groups, instead of working with the multiplication tables directly. This is
responsible for much of the nearly-linear time complexity of the various group
theoretic routines upon which we build our algorithm for Theorem~\ref{thm:main}.

While we provide a self-contained proof, Theorem~\ref{thm:isgrp} is an example
of a general approach we are developing for shifting promise problems to
deterministic problems, see also Section \ref{sec:outlook}. Promise problems are
especially common whenever inputs are given by compact encodings such as
black-box inputs, see Goldreich \cite{promise}. In ongoing work \cite{glassbox},
we introduce a more general process for verifying promises by specifying inputs
not as strings for a Turing machine, but rather as types in a sufficiently
strong Type Theory. Theorem~\ref{thm:isgrp} can be interpreted as an example of
such an input where the rows of the multiplication table are themselves treated
as inhabitants of a permutation type. The algorithm then effectively type-checks
that these rows satisfy the required introduction rules for a permutation group
type.  Type-checking is not in general decidable so the effort is to confirm an
efficient complexity for specific settings.  As a by-product of such models, the
subsequent algorithms also profit from using these rich data types; for more
details on this topic we refer to~\cite{glassbox}.

Group isomorphism has been such a tenacious problem that it has benefited from
analysis in stronger computational models, such as a RAM model where the data
is pre-loaded into registers and operations are costed as $O(1)$-time; cf.\
Section \ref{secCM}.  It is also common to use the promise hierarchy where
the axioms of a group and any further restrictions on the input are presupposed
without being part of the timing.  Recasting our result in similar models yields
the following.

\begin{corollary}
  Working in the promise hierarchy and using a RAM model with $O(1)$-time table
  look-ups, there is a deterministic algorithm that decides in time
  $O(n^{1+o(1)})$ whether two multiplication tables on $[n]$, with $n\in\ourset$,
  describe isomorphic groups.
\end{corollary}

\subsection{Structure}
In Section \ref{sec:prel} we introduce relevant notation and state
Theorems~\ref{thm:dense}--\ref{thm:iso},  which are the main ingredients for our
proof of Theorem~\ref{thm:main}. Since the proofs of
Theorems~\ref{thm:dense}--\ref{thm:iso} are more involved and partly depend on
technical Group Theory results,  we delay them until Appendix~\ref{app:proofs}. In
Section~\ref{sec:found} we discuss some algorithmic results required for our
proof of  Theorem \ref{thm:iso}. Proofs for our  main results are provided in
Section \ref{sec:proofsmain}. A conclusion and outlook are given
Section~\ref{sec:outlook}.

\section{Notation and preliminary results}\label{sec:prel} 

\subsection{Computational model}\label{secCM}
Throughout $n$ is the order of the multiplication tables used as input to
programs, so input lengths are $\Theta(n^2\log n)$. We write $\tilde{O}(n^d)$
for $O(n^d(\log n)^c)$, where $c$ and $d$ are constants, and we note that
$O(n(\log n)^{O((\log \log n)^c)})\subset O(n^{1+o(1)})$. Computations
are carried out on a multi-tape Turing machine (TM) with separate tapes for each input
group, an output tape, and a pair of spare tapes of length $O(n\log n)$ to store
our associated permutation group representations developed in the course of
Theorem~\ref{thm:isgrp}. In this model, carrying out a group multiplication
requires one to reposition the tape head to the correct product, at the cost of
$O(n^2\log n)$. That will be prohibitive for our given timing,  so
our first order of business will be to replace the input with an efficient
$\tilde{O}(n)$-time multiplication, cf.\ Remark \ref{rem:fast-multi}.  For
comparison, work of Vikas \cite{vikas} and Kavitha \cite{kavitha}  provide isomorphism tests
for abelian groups using $O(n)$ group operations.  Such an algorithm can be
considered as linear time in a random access memory (RAM) model where group and
arithmetic operations are stated as a unit cost.  This is partly how it is
possible to produce a running time shorter than the input length.  In general,
an $f(n)$-time algorithm in a RAM model produces an $O(f(n)^3)$-time algorithm
on a TM, see \cite[Section~2]{Papa}, although a lower complexity reduction may
exist for specific programs.  In particular results in \cite{das} \cite{kavitha}
are no more than $\tilde{O}(n^3)$-time on a Turing Machine.

\subsection{Group theory preliminaries} 
 We follow most common conventions in group theory, e.g.\ as in
\cite{rob,Seress}.  Given a group $G$, a subgroup $H\leq G$ is a nonempty subset
which is a group with the inherited operations from $G$.   Homomorphisms between
groups $G$ and $K$ are functions $f\colon G\to K$ with $f(xy)=f(x)f(y)$ for all
$x,y\in G$. For $S\subseteq G$, let $\langle S\rangle$ be the intersection of
all subgroups containing $S$; it is the smallest subgroup of $G$ containing $S$,
also called the subgroup \emph{generated} by $S$.  The commutator of group
elements $x,y$ is $[x,y]=x^{-1} y^{-1} xy$, and conjugation is written as
$x^y=x[x,y]$. For $X,Y\subseteq G$, let $[X,Y]=\langle [x,y] :  x\in X,y\in Y
\rangle$. The number of elements in $G$, the \emph{order} of $G$,  is denoted
$|G|$; in this work, $G$ always is a finite group.

Given a set $\pi$ of primes, a subgroup $H\leq G$ is  a \emph{Hall
$\pi$-subgroup} if every prime divisor of $|H|$ lies in $\pi$, and every prime
divisor of the  \emph{index} $|G:H|=|G|/|H|$ does not lie in $\pi$.  If
$\pi=\{p\}$, then  $H$ is a \emph{Sylow $p$-subgroup}.  A further convention is
to let $\pi'$ denoted the complement of $\pi$ in the set of all primes and to
speak of Hall $\pi'$-subgroups. The \emph{$\pi$-factorization} of an integer
$n>1$ is $n=ab$, where every prime divisor of $b$ lies in $\pi$, and no prime
divisor of $a$ lies in $\pi$.

A subgroup $B$ is \emph{normal} in $G$, denoted $B\unlhd G$, if $[G,B]\leq B$.
The group $G$ is \emph{simple} if its only normal subgroups are $\{1\}$ and $G$.
A \emph{composition series} for $G$ is a  series
$G=G_1>\ldots>G_m=\{1\}$ of subgroups with each $G_{i+1}\unlhd G_i$ and each
\emph{composition factor} $G_i/G_{i+1}$ is simple. The group $G$ is
\emph{solvable} if every composition factor is abelian. 

A normal subgroup $B\unlhd G$  \emph{splits} in $G$, denoted  $G=H\ltimes B$,
if there is  $H\leq G$ with $H\cap B=\{1\}$ and $G=\langle H,B\rangle$. Let
$\mathrm{Aut}(B)$ be the set of invertible homomorphisms $B\to B$.  
If $G=H\ltimes B$ and $h\in H$, then conjugation $b\mapsto b^h$  defines
$\theta(h)\in\mathrm{Aut}(B)$, and $\theta\colon H\to \mathrm{Aut}(B)$,
$h\mapsto \theta(h)$, is a homomorphism. Conversely,  given  $(H,B,\theta)$ with
homomorphism $\theta\colon H\to\mathrm{Aut}(B)$, there is a group
$H\ltimes_{\theta} B$ on the set $\{(h,b): h\in H, b\in B\}$ with product
$(h_1,b_1)(h_2, b_2)  = (h_1 h_2, b_1^{h_2} b_2)$; here we abbreviate
$b_1^{h_2}=\theta(h_2)(b_1)$. If conjugation in $G=H\ltimes B$ induces
$\theta\colon H\to \mathrm{Aut}(B)$, then $G$ is isomorphic to
$H\ltimes_{\theta} B$. In fact, the following observation holds; we will use this
lemma only in the situation that $B$ and $\tilde B$ are cyclic groups; this case of
Lemma \ref{lemcomp} is proved in \cite[Lemma~2.8]{cgroups}.

\begin{lemma}\label{lemcomp}
    Let $G=H\ltimes_{\theta}B$ and $\tilde{G}=\tilde{H}\ltimes_{\tilde{\theta}}\tilde{B}$.
    If $\alpha\colon H\to \tilde{H}$ and $\beta\colon B\to \tilde{B}$ are isomorphisms such 
    that for all $h\in H$ 
    \begin{align}\label{def:compatible}
        \tilde{\theta}(\alpha(h)) & = \beta\circ \theta(h)\circ \beta^{-1},
    \end{align}
    then $(h,b)\mapsto (\alpha(h),\beta(b))$ is an isomorphism $G\to \tilde{G}$.
    Conversely, if $G$ and $\tilde{G}$ are isomorphic and $H$ and $B$ have coprime orders, 
     then there is an isomorphism $G\to \tilde{G}$ of this form.
\end{lemma}

A variation of this observation also occurs in  \cite[Section 4.2]{BQ} and in
\cite[Section A.4.1]{GQ}. Note that in the last part of the lemma (when $H$ and
$B$ have coprime orders), the group $B$ is generated by all elements that have
order coprime to $|H|$; this makes $B$ a characteristic Hall subgroup.



\subsection{Number theory preliminaries}\label{secnt}
Our algorithm for Theorem \ref{thm:main} depends on crucial number theoretic
observations. For integers $n$, we characterize a family of prime divisors we
call \emph{strongly isolated}, and we show that any group $G$ of order $n$ in
our dense set $\ourset$ decomposes as $G=H\ltimes B$ such that the prime
factors of $|B|$ are exactly the strongly isolated prime divisors of $n$ that
are larger than $\log\log n$; we also prove that $B$ is cyclic. This
reduces our isomorphism test to considering the data $(H,B,\theta)$ and Lemma
\ref{lemcomp}. Not just isomorphism testing, but many natural questions of
finite groups reduce to properties of $(H,B,\theta)$, and so this decomposition
has interesting implications for computing with groups generally. Note that for
a group $G$ with order $n\in\ourset$, the decomposition described above defines
integers $a=|H|$ and $b=|B|$ that depend only on $n$. Our definition of
$\ourset$ will also imply that $b$ is square-free, and if a prime power $p^e$ with
$e>1$ divides $n$, then $p^e\leq \log n$ and $p^e\mid a$.  With $B$ being
cyclic, the group theory of $B$ is elementary, and while the group theory of $H$
can be quite complex, we will see in Theorem~\ref{thm:split}  that $H$ has
relatively small size, meaning that  brute-force becomes an efficient solution.
The following definition is central for our work.

\begin{definition}\label{def:isolated} 
  Let $n\in\mathbb{N}$. Write $2^{\nu_2(n)}$ for the largest $2$-power dividing
  $n$.  A prime $p\mid n$ is \emph{isolated} if $k=0$ for every prime power
  $q^k$ with  $q^k\mid n$ and $p\mid (q^k-1)$. If, in addition, $p\nmid |T|$
  for every non-abelian simple group $T$ of order dividing $n$, then $p$ is
  \emph{strongly isolated}. We write $\ppi(n)$ for the set of strongly isolated
  prime divisors of $n$; the subset of `large'  prime divisors is \[\bppi(n)=\{p\in\ppi(n) : p>\log\log n\}.\]
\end{definition}

For example, $31$ is isolated in $2^4 \cdot 5^2\cdot 31$ but not in $2^5\cdot
5^2\cdot 31$ or in $2^4\cdot 5^3\cdot 31$. As indicated above, the properties of
our dense set $\Upsilon$ are a critical ingredient in our algorithm  for Theorem
\ref{thm:main}; we are now in the situation to give the formal definition, cf.\
\cite{ErdosPalfy}.

\begin{definition}\label{def:ups} 
  Let $\ourset\subseteq\mathbb{N}$ be the set
  of all integers $n$ that factor as $n=ab$ such that:
    \begin{itemize} 
      \item[a)] if $p\mid a$ is a prime divisor, then $p\leq \log \log n$ and, if $p^e\mid
      a$, then  $p^e\leq \log n$;
      \item[b)] if $p\mid b$ is a prime divisor, then $p>\log\log n$ and $p\mid n$ is
       isolated;
      \item[c)] the factor $b$ is square-free, that is, if $p^e\mid b$ is a prime
        power divisor, then $e\in\{0,1\}$;
      \item[d)] the factor $b$ has at most $2\log \log n$ prime divisors.  
    \end{itemize} 
\end{definition}

\begin{theorem}\label{thm:dense}
  The set $\ourset$ is a dense subset of $\mathbb{N}$.
\end{theorem}

A proof of Theorem \ref{thm:dense} and more details on  $\ourset$ are
given in Appendix \ref{app:dense} and Section~\ref{sec:outlook}.

\subsection{Splitting results}
As mentioned in Section \ref{secnt}, the properties  of $n\in \ourset$ impose
limits on the structure of groups of order $n$; we prove the next theorem  in
Appendix~\ref{appthm:split}:

\begin{theorem}\label{thm:split}  
 Every group $G$ of order $n\in\ourset$ has  a unique Hall $\bppi(n)$-subgroup
 $B$, which is cyclic, and $G=H\ltimes B$ for some subgroup $H\leq G$ of small
 order $|H|\in (\log n)^{O((\log\log n)^c)}$ for some $c$.
\end{theorem}

\begin{remark}\label{remGen}{\rm In fact, our proof of Theorem \ref{thm:split}
also shows the following result for any group $G$ of any order $n\in
\mathbb{N}$: If $G$ is solvable  and $p\mid n$ is an isolated prime, or if $G$
is non-solvable, $p\mid n$ is a strongly isolated prime, and $p>\nu_2(n)$, then
$G$ has a normal Sylow $p$-subgroup $S\unlhd G$; in both cases,  $G=H\ltimes S$
for some $H\leq G$ by the Schur-Zassenhaus Theorem \cite[(9.1.2)]{rob}. Property
d) of Definition \ref{def:ups} is solely required to bound the size of
$H$.}
\end{remark} 

\begin{remark}
Theorem \ref{thm:split} shows that every group of order $n\in\ourset$ has a
`large' normal cyclic Hall subgroup that admits a `small' complement
subgroup.  This situation is covered by the main results of \cite{BQ}
(stated for groups with abelian Sylow towers), which deal with the case that a
normal Hall subgroup is abelian and the complement subgroups admit
polynomial-time isomorphism tests, see the discussion around \cite[Theorem
1.2]{BQ}. In conclusion, \cite{BQ} together with our Theorem \ref{thm:split}
already prove a polynomial-time isomorphism test for group of order
$n\in\ourset$. A significant part of this paper is devoted to  proving Theorem
\ref{thm:split}, and a careful analysis provides the nearly-linear time
complexity stated in Theorem \ref{thm:main}.
\end{remark}

The next theorem shows that we can construct generators for the decomposition in
Theorem~\ref{thm:split}; we discuss the proof of Theorem~\ref{thm:iso} in
Appendix~\ref{appthm:split}.

\begin{theorem}\label{thm:iso} 
  There is an $\tilde{O}(n^{1+o(1)})$-time algorithm that, given a group $G$ of order
  $n\in\ourset$, returns generators for $H,B\leq G$ such that $G=H\ltimes B$ and
  $B$ is a Hall $\bppi(n)$-subgroup. 
\end{theorem}

\section{Algorithmic preliminaries: presentations and complements}\label{sec:found}

\noindent We assume now that our input has been pre-processed and
confirmed to be a group by our algorithm for Theorem~\ref{thm:isgrp}. In so
doing, we also produce a series of important data types and accompanying
routines, including the following for each input group of order $n$:
\begin{itemize}
  \item a permutation group representation, 
  \item a generating set of size $O(\log n)$, 
  \item an algorithm to write group elements as a 
  product of the generators in time $\tilde{O}(n)$,
  \item an algorithm to multiply two group elements in time in $(\log n)^{O(1)}$, and
  \item an algorithm to test equality of elements in the group in time $\tilde{O}(n)$.
\end{itemize}
Remark~\ref{rem:fast-multi} explains some of these routines in more detail. The
advantage is that we can now multiply group elements without moving the Turing
machine head over the original Cayley tables which would cost us $\tilde O(n^2)$ steps
for each group operation.

Many of the following observations are variations on classical techniques
designed originally for permutation groups, e.g.\ as in
\cite{Luks:northeastern,Seress,simsbook}. We include proofs here to
demonstrate nearly-linear time when applied to the Cayley table model. For
simplicity we assume that all generating sets contain $1$. All our
lists have size $O(n)$, so all searches can be done
in nearly-linear time.

A \emph{membership test} for a subgroup $B\leq G$ is a function that, given
$g\in G$, decides whether $g\in B$. For example, a membership test for the
center $Z(G)\leq G$ is to report the outcome of the test whether $g\in G$
satisfies $g*s=s*g$ for all  $s\in S$. This test defines  $Z(G)$ without having
specified a generating set for it.

The next result discusses the construction of the Schreier coset graph
$\Cay(S,B)$ describing the action of $G=\langle S\rangle$ on cosets of $B\leq
G$: this is a labeled graph with vertices $\{xB : x\in G\}$, and an edge
$(uB,vB)$ exists if and only if $vB=suB$ for some $s\in S$. One such $s$ is
chosen as the label of $(uB,vB)$, and we denote such a labeled edge by
$(uB,svB;s)$. We note that a coset $uB$ is stored by a chosen representative
$u$. Note that if $B=1$, then this graph is the usual Cayley graph with respect
to the given generating set.

\enlargethispage{3ex}

\medskip

\mydes{SchreierGraph}
    {Given}{
        a group $G$ generated by $S\subseteq G$ and a membership test for a subgroup $B$,
    }
    {Return}{
        the Schreier coset graph $\Cay(S,B)=\{(xB,sxB;s)\mid s\in S, x\in G\}$ and together with a spanning 
        tree (a so-called \emph{Schreier tree}).
    }
  
\medskip

The next proposition uses well-known ideas that can already be found in
\cite{Luks:northeastern,Seress}; we provide a proof to justify our complexity
statements.

\begin{proposition}\label{prop:CayleyGraph}
  Let $G=\langle S\rangle$ with $|S|\in O(\log n)$ be a group of order $n$ and
  let $B\leq G$ with $[G:B]\in (\log n)^{O((\log \log n)^c)}$ be given by an
  $\tilde O(n^{1+o(1)})$-time membership test. {\rm\textsc{SchreierGraph}} can
  be solved in time $\tilde O(n^{1+o(1)})$. Once solved, the following hold:
  There is an $\tilde O(n^{1+o(1)})$-time algorithm that, given $g\in G$, finds
  a word $\bar{g}$ in  $S$ with $\bar{g} B=gB$. There is also an  an $\tilde
  O(n^{1+o(1)})$-time algorithm to compute a generating set for $B$ of size
  $(\log n)^{O((\log \log n)^c)}$.
\end{proposition}
\begin{proof}
As a pre-processing step, we can assume that for each of the $O(\log n)$
  generators in $S$ we have also stored its inverse; the inverse of $s\in S$ can
  be computed as $s^{n-1}$ using fast exponentiation. We initialize  graphs $C$
  and $T$, both with vertex set $V=\{B\}$ and empty edge sets, and proceed as
  follows:   While there is $s\in S$ and a vertex $xB\in V$ such that $sxB\notin
  V$, add   $sxB$ to $V$ and add $(xB,sxB;s)$ to the edge sets of $C$ and $T$.
  Otherwise, $sxB\in V$ and  we add $(xB,sxB;s)$ only to the edge set of $C$.
  This algorithm terminates after $O(|S||G:B|)$ steps, and then $|V|=|G:B|$ and
  $V=\Cay(S,B)$ with Schreier tree  $T$; this follows from a discussion of the
  orbit-stabiliser algorithm, for example in  \cite[Section~4]{Seress}. 

This iterative process guarantees that every vertex $gB$ we consider in the
algorithm (that is, $xB$ or $sxB$ with $xB\in V$ and $s\in S$) is represented by
a product of generators in $S$, say $gB=s_{i_t}\cdots s_{i_1}B$ with each
$s_{i_j}\in S$. In particular, we can assume that, along the way, we have also
stored the product $s_{i_1}^{-1}\cdots s_{i_t}^{-1}$ with the vertex. This
allows us to compare $gB$ with $yB\in V$ by  deciding whether the product of
$(s_{i_1}^{-1}\cdots s_{i_t}^{-1})$ with $y$ lies in $B$ via the membership
test. This shows that the algorithm described so far requires $O(|S||G:B|^2)$
group multiplications and membership tests.  

Note that we add a vertex $yB$ to $V$ if and only if we add an edge to the
Schreier tree $T$. By construction, the representative chosen to store $yB$ is
the product of the labels $s_{i_t},\ldots,s_{i_1}$ of the unique path
$B\overset{s_{i_1}}{\to}s_{i_1}B\overset{s_{i_2}}{\to}\cdots
\overset{s_{i_r}}{\to} s_{i_r}\cdots s_{i_1}B$ from $B$ to $yB$ in the Schreier
tree. In particular, it follows that $yB=\bar y B$ where $\bar y=s_{i_1}\cdots
s_{i_t}$. Now if $g\in G$ is given, then we first run over $V$ to identify the
vertex $y B\in V$ with $gB=yB$; as shown above, since we assume that inverses of
vertex representatives are stored, we can find $yB$ by using $O(|G:B|)$ group
multiplications and membership test applications. Subsequently, we look for the
path from $B$ to $yB$ in the Schreier tree; this can be done in time $O(|G:B|)$.
The labels of this path express the element $\bar y$ as a product of the
generators; by construction, $gB=yB=\bar yB$, and we define $\bar g=\bar y$.

It remains to prove the last claim. Let $\mathcal{T}$ be the set of
  representatives in $G$ describing the vertices in $V$; note that
  $\mathcal{T}=\{\bar y  : yB\in V\}$. This is is a transversal for $B$ in $G$,
  that is, $G$ is the disjoint union of all $tB$ with $t\in\mathcal{T}$.
  Schreier's lemma \cite[Lemma 4.2.1]{Seress} shows that $B$ is generated by the
  set of all $(\overline{st})^{-1}st$ where $s\in S$ and $t\in\mathcal{T}$. The
  previous paragraph shows that for each such $s$ and $t$ we can compute
  $\overline{st}$ in time $O(|G:B|)$. Since $\overline{st}$ is the stored
  representative of the vertex $stB$,  we have already stored its inverse
  $(\overline{st})^{-1}$. This shows that we can compute a set of Schreier
  generators for $B$ using $O(|S||G:B|^2)$ group multiplications and membership
  tests.
\end{proof}

\begin{corollary}\label{corB}  
  Let $G=\langle S\rangle$ be a group of order $n\in\ourset$ with $|S|\in O(\log
  n)$. If $G$ has a normal Hall subgroup $B$ with $|G/B|\in (\log
  n)^{O((\log\log n)^2)}$, then there is an $\tilde{O}(n^{1+o(1)})$-time
  algorithm that returns generators for $B$, and a membership test for $B$ that
  decides in time $\tilde O(n)$.
\end{corollary}
\begin{proof}
  Since a normal Hall subgroup is characteristic, hence unique,  $g\in G$ lies
  in $B$ if and only if the order of $g$ divides $b=|B|$; therefore we may test
  membership in $B$ by testing if $g^b=1$. This can be done in $O(\log b)$ group
  products using fast exponentiation, followed by a comparison with the identity
  $1$.  From our forgoing assumptions on $G$, this can be done in time
  $\tilde{O}(n)$. Finally, as $|S|\in O(\log n)$, generators of $B$ can be
  obtained from $\tilde{O}(|G:B|^2)$ group products and membership tests, using
  \textsc{SchreierGraph}$(S,B)$; all this can be done in time
  $\tilde{O}(n^{1+o(1)})$.
\end{proof}

We now introduce a tool that lets us find a complement to a normal Hall
$\bppi(n)$-subgroup.

\medskip
\mydes{BigSplit}
    {Given}{
        a group $G$ of order $n\in \ourset$,
    }
    {Return}{
        a subgroup $H\leq G$ such that $G=H\ltimes B$, where $B$ is the Hall 
        $\bppi(n)$-subgroup.
    }

\medskip 

To solve {\rm\textsc{BigSplit}} we need a brief detour into a generic model for
encoding groups via  presentations, cf.\ \cite[Section 1.4]{simsbook}:  
The \emph{free group} $F[X]$ on a given alphabet $X$ is formed by  creating a
disjoint copy $X^{-}$ of the alphabet and treating the elements of $F[X]$ as
words over the disjoint union $X\sqcup X^{-}$, including the empty word $1\notin
X\sqcup X^-$. Formally, one can replace the given set $X$ by $\{(x,1) : x\in
X\}$ and then defines $X^{-}=\{(x,-1) : x\in X\}$; to keep the notation simply,
we identify $x=(x,1)$ and $x^-=(x,-1)$. The empty word serves as the identity
and word concatenation is the group product; to impose the existence of inverses
we apply rewriting rules $xx^{-}\to 1$ and $x^{-}x\to 1$ for each  $x\in X$ and
corresponding $x^{-}\in X^{-}$. For a set $M$ we denote by  $M^X$ the set of
maps from $X$ to $M$, whose elements are naturally represented  as tuples
$\mathbf{m}=(\mathbf{m}_x)_{x\in X}$. For a group $G$, tuple $\mathbf{g}\in
G^X$, and word $w\in F[X]$, we assign an element $w(\mathbf{g})\in G$ by
replacing each variable $x^{\pm}$ in $w$ with the value $\mathbf{g}_x\in G$ and
$\mathbf{g}_x^{-1}\in G$, respectively, and then evaluating the corresponding
product in $G$.  The mapping $w\mapsto w(\mathbf{g})$ is a homomorphism
$\hat{\mathbf{g}}\colon F[X]\to G$, whose \emph{kernel}
$\ker\hat{\mathbf{g}}=\{w\in F[X]: w(\mathbf{g})=1\}$ is a normal subgroup of
$F[X]$. If $G$ is generated by the image $S=\{\mathbf{g}_x: x\in X\}$ and $R$
generates $\ker\hat{\mathbf{g}}$ as a normal subgroup, then the pair $\langle
S\mid R\rangle$ is a \emph{presentation} of $G$, where $R$ is a set of
\emph{relations} for $G$ relative to $S$.  Note that $\langle S\mid R\rangle$
carries all the information necessary to describe $G$ up to isomorphism;
however, in such an encoding isomorphism testing may be even become undecidable,
see \cite[Section~1.9]{simsbook}.

Our interest in presentations is to produce a relatively small number of
equations whose solutions help to solve {\rm\textsc{BigSplit}}; for that purpose
the following will suffice.

\begin{proposition}\label{prop:pres} 
  Let $G$ be a group of order $n\in\ourset$ with  Hall $\bppi(n)$-subgroup
  $B$. There is an $\tilde O(n^{1+o(1)})$-time
  algorithm to compute a presentation
  $\langle S\mid R\rangle$ of  $G/B$ such that  $|S|\in O(\log n)$
  and $|R|\in (\log n)^{O((\log\log n)^2)}$, and each $w\in R$ is a word in $S$
  of length $O(\log n)$. 
\end{proposition}
\begin{proof}
  As shown above, we find $G=\langle S\rangle$ with $|S|\in O(\log n)$.   Use {\rm
  \textsc{SchreierGraph}} and Corollary~\ref{corB} to get a transversal
  $\mathcal{T}$ for $B$ in $G$,  generators for $B$, and a rewriting algorithm
  that given $g\in G$, finds a word $\bar{g}$ in $S$ with $\bar{g}B=gB$. Choose
  a set $X=\{x_g : g\in S\}$ of variables, and for each $g\in G$ define 
  $w_{g}\in F[X]$ as the word in $X\cup X^{-}$ produced by replacing each $u\in S$ in
  $\overline{g}$ with~$x_{u}$.  Now $R=\{w_t x_s w_{ts}^{-1} : t\in \mathcal{T},
  s\in S\}$ is a set of relations for $G/B$ relative to $\{sB : s\in S\}$, cf.\ 
  \cite[Exercise 5.2]{Seress}. Note that $|R|\leq |G:B|\cdot |S|$, and $|G:B|
  \in  (\log n)^{O((\log\log n)^2)}$ by Theorem~\ref{thm:split}. 
  Also computing $w_g$ is dominated by the time $\tilde{O}(|G:B|^2)$ it takes 
  to compute $\bar{g}$. We do this on $O(|S||G:B|)$ elements for a total 
  time of $(\log n)^{O((\log \log n)^2)}\in O(n^{o(1)})$.

  Our relators have length $O(|G:B|)$, but \cite[Lemma~4.4.2]{Seress} yields  an
  $\tilde{O}(n)$-time algorithm to replace such relators with ones of length
  $O(\log n)$, by constructing a \emph{shallow} Schreier tree and \emph{short}
  Schreier generators. The analysis in the proof of \cite[Lemma 4.4.2]{Seress}
  shows that all this can be done using $\tilde O(|G:B|^2)$ group
  multiplications and membership tests; the claim follows.  For more details
  justifying our complexity statement we refer to our proof of Theorem
  \ref{thm:isgrp}: there we provide additional details on an explicit
  application of \cite[Lemma 4.4.2]{Seress}.
\end{proof}
 
\begin{remark}  
  Babai-Luks-Seress, Kantor-Luks-Marks and others (see the bibliography in
  \cite{KLM} and \cite[Section~6]{Seress}) developed various algorithms to
  construct presentations of (quotient) groups of permutations.  Their
  complexities range from polynomial-time, to polylogarithmic-parallel (NC), to
  Monte-Carlo nearly-linear time, and they produce presentations that can be
  considerably smaller than what we obtain in Proposition \ref{prop:pres}.
  Though it is not necessary for our complexity goals, we expect that a better
  analysis and a better performing implementation would use
  such methods instead of our above brute-force approach. 
\end{remark}

\begin{proposition}\label{prop:bigsplit}
  {\rm\textsc{BigSplit}} is in time  $\tilde O(n^{1+o(1)})$.
\end{proposition}
\begin{proof}
  \textsc{BigSplit} is solved via the function \textsc{Complement} discussed in
  \cite[Section 3.3]{KLM}; we briefly sketch the approach. Let $G=\langle
  S\rangle$ with $S=\{s_1,\ldots,s_d\}$ and $d\in O(\log n)$. Use the algorithm of
  Proposition~\ref{prop:pres} to get a  presentation $\langle x_1,\ldots,x_d\mid
  R\rangle$  for $G/B$, such that each $x_i=x_{s_i}$ as defined in the proof of
  Proposition \ref{prop:pres}.   Every complement $H$ to $B$, if it exists, is
  generated by $\{s_1m_1,\ldots,s_d m_d\}$ for some $m_1,\ldots,m_d\in B$, and
  such a generating set satisfies the relations in $R$, cf.\
  \cite[(2.2.1)]{rob}. We attempt to compute $m_1,\ldots,m_d$ by solving the
  system of equations resulting from $w(s_1m_1,\ldots, s_d m_d) =1$ with $w$
  running over $R$: recall that $w(s_1m_1,\ldots, s_dm_d)\in G$  is defined as
  $w(\mathbf{g})$ with $\mathbf{g}=(s_1m_1,\ldots,s_dm_d)$, see the remarks
  above Proposition \ref{prop:pres}. A complement to $B$ exists if and only if
  this equation system has a solution. Since each $w(s_1,\ldots,s_d)$ lies in
  the finite cyclic group $B$, this system can be described by an integral
  matrix with $\log n$ variables and  $(\log n)^{O((\log\log n)^2)}$ equations;
  using the algorithms of \cite{hermite}, it can be solved via  Hermite Normal
  Forms in time $(\log n)^{O((\log\log n)^2)}$.
\end{proof}

\section{Proofs of the main results}\label{sec:proofsmain}   

\subsection{Proof of Theorem~\ref{thm:main}: isomorphism testing}\label{appthm:main}

\begin{proof}[Proof of Theorem \ref{thm:main}]

  Given two binary maps $[n]\times[n]\to [n]$, we decide that $n\in\ourset$ in
  time $O(n)$, and we use Theorem \ref{thm:isgrp} to decide in time $\tilde
  O(n^2)$ whether these maps describe Cayley tables. If so, we have been given
  two groups $G$ and $\tilde G$ of order $n\in\ourset$, and we can use Theorem
  \ref{thm:iso} to find, in time $\tilde O(n^{1+o(1)})$, generators for
  subgroups $H,B\leq G$ and $\tilde H,\tilde B\leq \tilde G$ with $G=H\ltimes B$
  and $\tilde G=\tilde H\ltimes\tilde B$. Having generators of these subgroups,
  we can define homomorphisms $\theta$ and $\tilde\theta$ such that
  $G=H\ltimes_\theta B$ and $\tilde G=\tilde H\ltimes_{\tilde \theta}\tilde B$.
  Since $n\in\ourset$, we know that $B$ and $\tilde B$ are cyclic, hence $B$ and
  $\tilde B$ are isomorphic if and only if $|B|=|\tilde B|$. Moreover, since
  $|H|,|\tilde H|\leq (\log n)^{O((\log \log n)^c)}$ we can test isomorphism
  $H\cong \tilde H$ using brute-force   
  methods in time  $(\log n)^{O((\log \log n)^d)}\subseteq \tilde O(n^{1+o(1)})$. If $H\cong \tilde H$ and 
  $B\cong \tilde B$ is established, then we can identify $H=\tilde H$ and $B=\tilde B$, 
  and test $G\cong \tilde G$ by using Lemma \ref{lemcomp}: since $B$ is cyclic, $\Aut(B)$
  is abelian, and so Condition \eqref{def:compatible} reduces to
  $\tilde\theta(\alpha(h))=\theta(h)$ for all $h\in H$, which we can test by
  enumerating $\Aut(H)$ and looking for a suitable $\alpha$; since  $|H|$ is
  small, such a brute-force enumeration is efficient and in $\tilde O(n^{1+o(1)})$.
\end{proof}

\begin{remark}\label{remUps} 
  We briefly explain how to decide $n\in\Upsilon$ on
  a Turing machine. Computing the table of primes between $1$ and $n$ on a
  Turing machine can be done in time $O(n(\log n)^2 \log\log n)$, see \cite[p.\
  227]{schonhage}, and the number of primes is clearly bounded by $n$.
  Approximating an upper bound for $\log n$ can be done in $O(n)$ by counting
  the bits that are required to represent $n$. That integer arithmetic for
  integers $n$ can be done in $(\log n)^{O(1)}$  on multi-tape Turing machines
  follows from \cite[Chapter III.6]{schonhage}. It follows that $n\in\Upsilon$
  can be decided in time $\tilde O(n^2)$.  
\end{remark}

\subsection{Proof of Theorem~\ref{thm:isgrp}:  recognizing groups}\label{app:isgroup}

Similar to \cite[Chapter 2]{latin}, our strategy for recognizing groups uses
Cayley's Theorem \cite[(1.6.8)]{rob}. The latter implies that the rows of an
$n\times n$ group table can be interpreted as permutations which form a
\emph{regular} permutation group on $[n]$, that is, the group is transitive on
$[n]$ and has trivial point-stabilizers, see \cite[Section~1.2.2]{Seress}.

A new idea is to exploit that groups of order $n$ can be
specified by generating sets of size $\log n$, so some $\log n$ rows determine
the entire table.  Once the input is verified to be a latin square, our approach
is to define a permutation group generated by $O(\log n)$ rows, and then compare
its Cayley table with the original table.  In more abstract terms, our algorithm
creates an instance of an abstract permutation group data type, as defined in
\cite[Section~3]{Seress}. That data type is guaranteed to be a group and
so the promise is converted into a computable type-check: We confirm that the
group we create in this new data type is the one specified by the original
table; the proof given below makes this argument specific. This methodology of
removing a promise by appealing to a type-checker generalizes; we refer to  our
forthcoming work \cite{glassbox} for more details.

\begin{proof}[Proof of Theorem  \ref{thm:isgrp}] 
Let $\ast\colon [n]\times[n]\to [n]$ be the multiplication defined by the
table~$T$; recall our assumption  that rows and columns are both labelled by
$1,2,\ldots,n$. If the table is not \emph{reduced}, that is, if the first row or
the first column do  not contain $1,2,\ldots,n$ in order, then ``$1$'' is not an
identity and we return false.  Running over the table, we can check in time
$\tilde O(n)$ whether any particular row contains distinct symbols only, see for
example \cite[p.\ 434]{pet}; once this is confirmed for all $n$ rows, the rows
of $T$ describe permutations of $[n]$. All this preprocessing can be done in
time $\tilde O(n^2)$. We verify in the proof below that also each column of $T$
contain distinct symbols only (and return false if this is not the case); then
$L=([n],\ast)$ describes a loop with identity $1$; recall that a loop is a
quasigroup (that is, a set with a binary operation whose multiplication table is
a Latin square) that has an identity.

  For $i\in[n]$ denote by $\lambda_i\in{\rm Sym}_n$ the map $[n]\to [n]$ defined
  by left multiplication $\lambda_i(a)=i\ast a$; since $T$ is reduced,
  $\lambda_i$ is given by the $i$-th row of $T$. We define \[\Lambda=\langle
  \lambda_i : i \in [n]\rangle;\] since each $\lambda_i(1)=i$, this is a
  transitive subgroup of ${\rm Sym}_n$. Note that if one of the columns of $T$
  contains duplicate symbols, then $\Lambda$ is not regular. It follows that $L$
  is a group if and only if $\Lambda$ is a regular permutation group on $[n]$,
  if and only if  $\lambda_i\lambda_j=\lambda_{i\ast j}$ for all $i,j\in[n]$,
  see \cite[Theorems~2.16 \& 2.17]{latin}.   Since $\Lambda$ is transitive, it
  is regular if and only if the stabiliser $\Lambda_1$ of $1\in[n]$ in $\Lambda$
  is trivial. We now show how to find a generating set of $\Lambda$ of size
  $O(\log n)$ and prove that $\Lambda$ is regular, or establish that $L$ is not
  a group and return false.

We first introduce some notation. For a subset  $S\subseteq[n]$ define
  \[\Lambda(S)=\langle\lambda_i : i \in S\rangle\leq \Lambda.\]  Let
  $\Lambda(S)(1)=\{\lambda(1) : \lambda\in \Lambda(S)\}$ and
  $\Lambda(S)_1=\{\lambda\in \Lambda(S) : \lambda(1)=1\}$ be the orbit and
  stabiliser of~$1$ in $\Lambda(S)$, respectively.  We describe how to use the
  orbit-stabiliser algorithm (cf.\ the proof of Proposition
  \ref{prop:CayleyGraph})  to get $\Lambda(S)(1)$, a generating set for
  $\Lambda(S)_1$ (so-called Schreier generators), and for each
  $x\in\Lambda(S)(1)$ an transversal element $\lambda_{(x)}\in\Lambda(S)$ with
  $\lambda_{(x)}(1)=x$. We proceed in two steps.

  The first step is to find a subset $S\subseteq [n]$ of size $O(\log n)$ such
  that $\Lambda(S)$ is transitive on $[n]$. We start by choosing a subset
  $S\subseteq[n]$ of size $O(\log n)$ and by   copying $\{\lambda_i : i \in S\}$
  to a separate tape; we can assume that $1\in S$. In the following we will
  mainly work with this short tape; scanning it  takes time $\tilde O(n)$. With
  this assumption, we can compute $\Lambda(S)(1)$ in time $\tilde O(n^2)$ by
  using the usual orbit enumeration of the orbit-stabiliser algorithm, see the
  proof of Proposition \ref{prop:CayleyGraph}. If $\Lambda(S)(1)\ne [n]$, then
  we find $j\in [n]\setminus \Lambda(S)(1)$, replace $S$ by $S\cup
  \{\lambda_j\}$, and iterate.  By  construction, $\lambda_j(1)=j$, so the new
  orbit and hence the new group $\Lambda(S)$ have increased in size. We can
  increase $|\Lambda(S)|$ only $O(\log n)$ times, so we repeat the work of the
  previous paragraph $O(\log n)$ times. In conclusion, in time $\tilde O(n^2)$
  we can find $S\subseteq [n]$ of size $O(\log n)$ with $\Lambda(S)$ transitive
  on $[n]$. 

The second step is to use the method described in the proof of
Proposition~\ref{prop:pres} (based on \cite[Lemma~4.4.2]{Seress}) to replace $S$
by a new set $S$ (also of size $O(\log n)$) that yields short transversal
elements (each of length $O(\log n)$) and short Schreier generators. To describe
the construction of \cite[Lemma~4.4.2]{Seress}, we need more notation. For
elements $g_1,\ldots,g_k\in \Lambda(S)$ with $k\in O(\log n)$ we define $C_k =
\{g_k^{e_k}\ldots g_1^{e_1} : \text{ each }e_i\in \{0,1\}\}$ and the
corresponding orbit $O_k=(C_kC_k^{-1})(1)$, and call $C_k$ nondegenerate if
$|C_k|=2^k$. As explained on \cite[p.\ 65]{Seress}, the set $O_k$ can be
constructed by starting with $\Delta_0=\{1\}$ and recursively defining
$\Delta_i=\Delta_{i-1}\cup \{h_i(c) : c \in \Delta_{i-1}\}$, where $h_i$ is the
$i$-th element in the list $g_k^{-1},\ldots,g_1^{-1},g_1,\ldots,g_k$. With this
construction, $O_k = \Delta_{2k}$. For every $x\in O_k$ constructed this way we
also store one tuple $[x; j_1,\ldots,j_\ell]$, meaning that
$\lambda_{(x)}=h_{j_1}\cdots h_{j_\ell}$ was one element we used to construct
$x=\lambda_{(x)}(1)$ in $O_k$. These elements $\lambda_{(x)}$ will be the new
short transversal elements; recall that  $\ell\in O(\log n)$ by construction. If
the elements $g_1,\ldots,g_k,g_1^{-1},\ldots,g_k^{-1}$ are stored on a separate
tape (which we can scan in time $\tilde O(n)$), the set $O_k$ can be constructed
in time $\tilde O(n^2)$; note, however, that we have not yet computes
$\lambda_{(x)}$, but rather stored the tuple  $[x;j_1,\ldots,j_\ell]$
describing~it.

  Now we follow the proof of \cite[Lemma~4.4.2]{Seress} to find
  $g_1,\ldots,g_k\in \Lambda(S)$ with $k\in O(\log n)$ such that $C_k$ is
  nondegenerate and $|O_k|=n$. We start by choosing a non-identity generator
  $g_1\in S$. If  $g_1,\ldots,g_i$ are found such that $C_i$ is nondegenerate,
  but $|O_i|<n$, then  there exists $j\in O_i$ and $s_r\in S$ such that
  $s_r(j)\notin O_i$. Now define $g_{i+1}=s_r\lambda_{(j)}$ (multiplied in time
  $\tilde O(n^2)$), where $\lambda_{(j)}$ is the stored transversal element for
  $j\in O_i$, and iterate until all required conditions are met. We have to do
  at most $O(\log n)$ iterations, and correctness follows from the proof of
  \cite[Lemma~4.4.2]{Seress}. In conclusion, in time $\tilde O(n^2)$ we have
  found a new generating set $\{g_1,\ldots,g_k,g_1^{-1},\ldots,g_k^{-1}\}$ of
  size $O(\log n)$ of the original group $\Lambda(S)$, and we have also
  constructed for each $x\in[n]$ a tuple $[x; j_1,\ldots,j_\ell]$ of length
  $O(\log n)$ encoding a short transversal element $\lambda_{(x)}$ mapping $1$
  to $x$. Let $\mathcal{T}$ be the set of all these tuples, and note that we can
  scan $\mathcal{T}$ in time $\tilde O(n)$.

  To simplify the notation, in the following we assume that the new generating
  set is again stored on a separate short tape and indexed by $S\subseteq [n]$
  so that  $\{g_1,\ldots,g_k,g_1^{-1},\ldots,g_k^{-1}\}=\{\lambda_s : s\in S\}$:
  using the short tape, we evaluate each $g_i^\pm(1)=u$ and then scan the
  original table $T$ in time $\tilde O(n^2)$ to confirm $g_i^\pm(1)=\lambda_u$
  and to add $u$ to $S$; if we find out that $g_i^\pm(1)\ne \lambda_u$, then
  $\Lambda(S)$ is not regular and we return false.

  We now construct a set $R$ of relations that  encode  Schreier generators for
  the stabiliser $\Lambda(S)_1$. For each $[x;j_1,\ldots,j_\ell]\in\mathcal{T}$
  and for each $r\in S$ evaluate $y=\lambda_r(x)$ and scan $\mathcal{T}$ for the
  entry $[y;v_1,\ldots,v_m]$; this means that $\lambda_{v_1}\cdots
  \lambda_{v_m}(1)=y=\lambda_r\lambda_{j_1}\ldots \lambda_{j_\ell}(1)$, and we
  add $[r,j_1,\ldots,j_\ell;v_1,\ldots,v_m]$ to $R$. The whole of $R$ can be
  construced in time $\tilde O(n^2)$; moreover, $R$ has $\tilde O(n)$ entries of
  length $O(\log n)$ each. 
  
Next, we use $R$ to check that the stabiliser $\Lambda(S)_1$ is trivial. It
follows from Schreier's Lemma \cite[Lemma 4.2.1]{Seress} that $\Lambda(S)_1$ is
generated by the Schreier generators encoded by the relations stored in $R$;
more precisely, the stabiliser is trivial  if and only if for each of the
$\tilde O(n)$ relations $[j_1,\ldots,j_\ell; v_1,\ldots,v_m]$ in $R$ we have
$\lambda_{j_1}\cdots \lambda_{j_\ell}=\lambda_{v_1}\cdots \lambda_{v_m}$. (At
this stage we only know that both sides of this equation map $1$ to the same
element.) Note that by construction $\ell,m\in O(\log n)$ and each generator
$\lambda_s$ is stored on a short tape that we can scan in time $\tilde O(n)$.
We now argue that we can multiply two permutations given by $n$-tuples in
time $\tilde O(n)$: recall that we represent a permutation $a\in {\rm
Sym}_n$ as the $n$-tuple $[1^a,\ldots,n^a]$; labelling these entries by
$1,\ldots,n$, and applying the aforementioned $\tilde O(n)$-time sorting
algorithm (see  \cite[p.\ 434]{pet}), we can compute the following
transformation (where the top row indicates the labels):\enlargethispage{2ex}
\[a=\left[\begin{array}{c|cccc} 
   i& 1,&2,&\ldots ,& n\\
   i^a&  1^a,&2^a,&\ldots,&n^a
  \end{array}\right] \quad\to\quad
  a=\left[\begin{array}{c|cccc}
   i& 1^{a^{-1}},&2^{a^{-1}},&\ldots ,& n^{a^{-1}}\\
    i^a& 1,&2,&\ldots,&n \end{array}\right];\]     
    this is done by simply sorting
    $[1^a,\ldots,n^a]$ and keeping track of the labels. Now if $b$ is a second
    permutation represented as $[1^b,\ldots,n^b]$, then $ab$ is the permutation
    represented by 
{\small
\[ab=\left[\begin{array}{c|cccc}
   i& 1^{a^{-1}},&2^{a^{-1}},&\ldots ,& n^{a^{-1}}\\
    i^a& 1,&2,&\ldots,&n
  \end{array}\right]\cdot
\left[\begin{array}{c|cccc}
   i& 1,&2,&\ldots ,& n\\
   i^b&  1^b,&2^b,&\ldots,&n^b
  \end{array}\right]
=
\left[\begin{array}{cccc}
    1^{a^{-1}},&2^{a^{-1}},&\ldots ,& n^{a^{-1}}\\
    1^b,&2^b,&\ldots,&n^b \end{array}\right]; \]} 
  sorting the right hand array
   by its first row (again in time $\tilde O(n)$), we obtain the $n$-tuple
   describing $ab$. In conclusion, in time $\tilde O(n)$ we can compute the
   elements  $\lambda_{j_1}\cdots \lambda_{j_\ell}$ and $\lambda_{v_1}\cdots
   \lambda_{v_m}$, and check for equality. Since $|R|\in \tilde O(n)$, checking
   that $\Lambda(S)_1$ is trivial can be done in time $\tilde O(n^2)$. If
   $\Lambda(S)_1$ is not trivial, then $\Lambda(S)$ is not regular and we return
   false. Otherwise, we have proved that $\Lambda(S)$ is a regular permutation
   group on $[n]$, contained in $\Lambda$.

It remains to show that $\Lambda\leq \Lambda(S)$. For this it is sufficient to
show that each $\lambda_x\in \Lambda(S)$; because $\Lambda(S)$ is regular, the
latter holds if and only if $\lambda_x=\lambda_{(x)}$ is the constructed
transversal element. For this we first sort $\mathcal{T}$ such that its elements
correspond to $\lambda_{(1)},\ldots,\lambda_{(n)}$; note  that $\mathcal{T}$ has
$n$ entries of length $O((\log n)^2)$ bits, so we can sort $\mathcal{T}$ in time
$\tilde O(n^2)$. Once sorted, we replace each entry $[x;j_1,\ldots,j_\ell]$ by
the permutation $\lambda_{(x)}$; the latter has already been constructed above
by multiplying $\lambda_{j_1}\cdots \lambda_{j_\ell}$. Lastly, in time $\tilde
O(n^2)$ we scan the original table $T$ and compare each
$\lambda_x=\lambda_{(x)}$. 
\end{proof} 

Some parts of the latter proof are similar to the proof
of Proposition \ref{prop:CayleyGraph}. One difference is that in the proof of
Proposition \ref{prop:CayleyGraph} checking whether an element $xB$ lies in the
orbit $V$ requires $|V|$ group multiplications and membership tests; in the
proof above, this only requires running over a list $\mathcal{T}$ of size $n$
and comparing numbers. In Proposition \ref{prop:CayleyGraph} the orbit is small
of size $|G:B|$, but applying the membership is more expensive; in the proof
above, the orbit is large of size $n$, but checking equality is more efficient.

\begin{remark}\label{rem:fast-multi} Having converted our Cayley table into a
regular permutation group $G$, we now have a generating set $S$ of size $O(\log
n)$ and a corresponding Schreier tree; both  are stored in a separate tape of
length $\tilde O(n)$.  This is the tape that is used whenever we operate in the
group; the original input tapes will never be revisited.   As mentioned
in the proof of Theorem \ref{thm:isgrp}, we may
assume that the Schreier trees are \emph{shallow}, that is, they have depth
bounded by $O(\log n)$, see \cite[Lemma~4.4.2]{Seress}. The algorithms we now
use are as follows. Each $g\in G$ is a node in the Schreier tree and there is a
unique path from the origin to $g$; if $g_1,\ldots,g_k\in S$ are the labels of
that path, then $g=g_k\cdots g_1$; note that $k\in O(\log n)$ since the Schreier
tree is shallow. We now describe how to compute the product of these labels.
Recall that $g_1,\ldots,g_k$ are permutations on $[n]$, so we can compute the
image of $1\in[n]$ under $g_k\cdots g_1$, by looking up $i_1=g_1(1)$,
$i_{2}=g_{2}(i_1)$, etc, until we obtain $u=g_k\cdots g_1(1)$. This scan occurs
on the short tape in time $\tilde O(n)$. Since the group is regular, there is a
unique $g\in G$ with $u=g(1)$, which  determines $g=g_k\cdots g_1$.  To multiply
elements $g_k\cdots g_1$ and $g_j'\cdots g_1'$, we merely concatenate the
generators, and continue with this word, yielding a $(\log n)^{O(1)}$-time
multiplication. We note that none of our product lengths exceed $(\log
n)^{O(1)}$. To compare $g_k\cdots g_1$ and $g_j'\cdots g_1'$ with $k,j\in O(\log
n)$, we determine and compare $g_k\cdots g_1(1)$ and $g_j'\cdots g_1'(1)$ in
time $\tilde{O}(n)$. More details of these methods are given in \cite[p.\
85--86]{Seress}.   
\end{remark}

\section{Conclusion and outlook}\label{sec:outlook}

\noindent We have shown that when restricted to a dense set $\ourset$ of group
orders, testing isomorphism of groups  of order $n\in\ourset$ given by Cayley
tables can be done in time $\tilde{O}(n^2)$; this significantly improves the
known general bound of $n^{O(\log n)}$. We note that $|\ourset\cap
\{1,2,\ldots,10^k\}|/10^k$ is  approximately  0.723, 0.732, and 0.786 for
$k=8,9,10$, respectively; to determine whether $n\in \upsilon$,  all logarithms
have been computed with respect to the basis $2$. Recall that  one can decide if
$n\in \ourset$ in time $\tilde{O}(n^2)$, see Remark~\ref{remUps}. We note that
our work \cite{PART2} considers isomorphism testing of cube-free groups, but
under the assumption that groups are given as permutation groups.

We have proved that groups of these orders admit a computable factorisation
$G=H\ltimes B$ with the following useful property: firstly, the
\emph{\textbf{H}ard group theory} of $G$ is captured in $H$, but $|H|$ is
\emph{small} compared to $|G|$ so brute-force methods can be applied to $H$;
secondly, the \emph{\textbf{B}ig number theory} of $|G|$ is captured by $|B|$,
but $B$ is cyclic, hence its  group theory is \emph{easy}. These decompositions
exist for a dense set of group orders, so we expect this will be useful for
other computational tasks as well. In fact, we will exploit properties of these
decompositions in future work:  This paper is part of our program  to enhance
group isomorphism, see \cite{cgroups,PART2} for recent work, and we plan to
extend the present results to other input models. Specifically, in our current
work \cite{glassbox} we develop a new black-box input model for groups (based on
Type Theory)  that does not need a \emph{promise} that the input really encodes
a group, so algorithms for this model  can be implemented within the usual
polynomial-time hierarchy. Due to Theorem~\ref{thm:isgrp}, the algorithms
presented here do also not require a promise that the input tables describe
groups. We conclude by mentioning that our algorithm for isomorphism testing can
be adapted to find a single isomorphism, generators for the set of all
isomorphisms, or to prescribe a canonical representative of the isomorphism type
of a single group.

\section*{Acknowledgments}

\noindent Both authors thank Joshua Grochow and Youming Qiao for comments on
earlier versions of this draft. They also thank the Newton Institute
(Cambridge, UK) where some of this research took place and we recognize the
funding of EPSRC Grant Number EP/R014604/1. Dietrich was supported by an
Australian Research Council grant, identifier DP190100317. Wilson was supported
by a Simons Foundation Grant, identifier \#636189.  Special thanks to Shyan Akmal
for referring us to the work of Rajagopalan--Schulman.


   \appendix

\section{Proofs of Theorems \ref{thm:dense}--\ref{thm:iso}}\label{app:proofs}

\subsection{Number theory: Proof of Theorem~\ref{thm:dense}}\label{app:dense}

\begin{proof}[Proof of Theorem~\ref{thm:dense}]
Erd\H{o}s-P\'alfy \cite[Lemma~3.5]{ErdosPalfy} showed that almost every $n\in\mathbb{N}$
	has the property that  if a prime $p>\log\log n$ divides $n$, then $p^2\nmid
	n$; thus,  $n=p_1^{e_1}\ldots p_k^{e_k}b$ with $b$ square-free, every prime
	divisor of $b$ is greater than $\log\log n$, and $p_1,\ldots,p_k\leq \log\log
	n$ are distinct primes. Let $x>0$ be an integer. We now compute an estimate
	for the number $N(x)$ of integers $0<n\leq x$ which are divisible by a prime
	$p\leq \log \log n$ such that the largest $p$-power $p^e$ dividing $n$
	satisfies $p^e>\log n$. We want to show that $N(x)/x \to 0$ for $x\to \infty$;
	this proves that for almost all integers $n$, if $p^e\mid n$ with $p\leq \log
	\log n$, then $p^e\leq \log n$. To get an upper bound for $N(x)$, we consider
	integers between $\sqrt{x}$ and $x$ with respect to  
	the above property, and add $\sqrt{x}$ for all integers between $1$ and
	$\sqrt{x}$. Note that if $p^e\geq \log n$, then $e\geq \log\log n / \log p$.
	Since we only consider $\sqrt{x}\leq n\leq x$, this yields $e\geq c(x)$ where
	$c(x)= \log\log \sqrt{x}/\log\log\log x$. Note that $c(x)\to \infty$ if $x\to
	\infty$, thus
	\begin{eqnarray*}
		N(x) &\leq & \sqrt{x} + \sum\nolimits_{k=2}^{\lfloor \log \log \sqrt{x}\rfloor}\frac{x}{k^{c(x)}}
		\;\;\leq\;\;  \sqrt{x} + x \int_2^{\log\log\sqrt{x}} \frac{1}{y^{c(x)}}\text{d}y\\
		&=& \sqrt{x} +  \frac{x}{1-c(x)}\left[\frac{1}{(\log\log\sqrt{x})^{c(x)-1}}-\frac{1}{2^{c(x)-1}}\right].
	\end{eqnarray*}
	Since $1/(1-c(x))\to 0$ from below, we can estimate:
	\begin{eqnarray*}
		N(x) 
		&\leq & \sqrt{x} +  x\left|\frac{1}{1-c(x)}\right|\left[-\frac{1}{(\log\log\sqrt{x})^{c(x)-1}}+\frac{1}{2^{c(x)-1}}\right]		
    \leq \sqrt{x} +  x \left|\frac{1}{1-c(x)}\right|\left[\frac{1}{2^{c(x)-1}}\right];
	\end{eqnarray*}
thus $N(x)=o(x)$, since $N(x)/x \leq \sqrt{x}/x +
  \left|1/(1-c(x))2^{c(x)-1}\right|\to 0$ if $x\to \infty$. This proves that the
  set $\ourset_1$ of all positive integers satisfying conditions a,c) in
  Definition \ref{def:ups} is dense.  By \cite[Lemmas 3.5 \& 3.6]{ErdosPalfy},
  the set $\ourset_2$ of positive integers $n$ satisfying conditions b,c) is
  dense as well. An inclusion-exclusion argument proves that
  $\ourset_3=\ourset_1\cap\ourset_2$ is  dense. The Hardy-Ramanujan Theorem
  \cite[Section 8]{hardy} proves that the set  $\Upsilon_4$ of integers $n$ that
  have at most $2\log \log n$ distinct prime divisors is dense, and now the
  intersection $\ourset_3\cap\ourset_4$ is dense as well. Since
  $\ourset_3\cap\ourset_4\subseteq \ourset$, the claim follows.
\end{proof}

\subsection{Splitting theorems: Proofs of Theorems~\ref{thm:split} \& \ref{thm:iso}}
\label{appthm:split}

\begin{proof}[Proof of Theorem~\ref{thm:split}]
Let $G$ be a group of order  $n\in\ourset$; we first show that $G$ has a normal
Hall $\bppi(n)$-subgroup, and  if $G$ is solvable, then
there is a normal Hall $\ppi(n)$-subgroup.

\medskip

First, let  $G$ be solvable. We show that $G$ has a normal Sylow
  $p$-subgroup for every $p\in \ppi(n)$. Let $q\ne p$ be a prime dividing $n$, and
  let $A$ be a Hall $\{p,q\}$-subgroup of $G$ of order $p^eq^f$; see
  \cite[Section 9.1]{rob}.  The Sylow Theorem \cite[(1.6.16)]{rob} shows that the  number $h_p$ of Sylow
  $p$-subgroups of $A$ divides $q^f$ (and hence $n$) and $p\mid (h_p-1)$.  Since
  $p\mid n$ is isolated, we have $h_p=1$ and  $A$ has a normal Sylow
  $p$-subgroup. Now fix a Sylow basis $\mathcal{P}=\{P_1,\dots,P_s\}$  for $G$,
  that is, a set of Sylow subgroups, one for each prime dividing $n$, such that
  $P_iP_j=P_jP_i$ for all $i$ and $j$; see \cite[Section~9.2]{rob}. Let $P=P_u$
  be the Sylow $p$-subgroup for $G$ in $\mathcal{P}$. Since $G=P_1\cdots P_s$,
  every $g\in G$ can be written as $g=g_1\ldots g_s$ with each $g_j\in P_j$.
  Since $PP_j=P_j P$, the group $PP_j$ is a Hall $\{p,p_j\}$-subgroup. As shown
  above, $P\unlhd PP_j$, so each $g_jP=Pg_j$. Thus,  $gP=g_1\ldots   g_s P =
  Pg_1\ldots g_s=Pg$, so $P\unlhd G$.

\medskip
  
 Second, suppose that $G$ is non-solvable. We show that $G$ has a normal Sylow
  $p$-subgroup for every $p\in\bppi(n)$. Being non-solvable, $G$  has a
  non-abelian simple composition factor, so $|G|$ is divisible by~$4$, see
  \cite[p.\ 155]{feit}. Since $n\in\ourset$, Definition \ref{def:ups}c)  implies that
  $2^{\nu_2(n)}\leq \log n$, so $\nu_2(n)\leq \log\log n<p$ for every
  $p\in\bppi(n)$.

  In the following we use the Babai-Beals filtration $G\geq {\rm PKer}(G)\geq
  {\rm Soc}^\ast(G)\geq O_\infty(G)\geq 1$, see \cite[Section 1.2]{babaibeals}:
  here $O_\infty(G)$ is the largest normal subgroup of $G$ and ${\rm
  Soc}^*(G)/O_\infty$ is the socle of $G/O_\infty(G)$, which is defined to be
  the subgroup generated by all minimal normal subgroups. This socle decomposes
  as  $T_1\times\cdots \times T_{\ell}$ where each $T_i$ is a non-abelian simple
  normal subgroup of $G/O_\infty(G)$,  see also \cite[pp.\ 157--159]{Seress}.
  The group  ${\rm PKer}(G)/O_\infty(G)$ is the kernel of the  permutation
  representation $G/O_\infty(G)\to {\rm Sym}_{\ell}$ induced by the conjugation
  action on $\{T_1,\ldots,T_{\ell}\}$.

 First, we claim that $p\nmid |G:O_{\infty}(G)|$. Note that $p\nmid |T_i|$ for each $i$ since
  $p\in\bppi(n)$. By the Classification of Finite Simple Groups, a prime $r$
  divides $|\Aut(T_i)|$ only if $r$ divides $|T_i|$, as seen from the list of
  known orders\footnote{The finite simple groups (Classification Theorem of
  Finite Simple Groups) are listed in \cite[Section 1.2]{wilson}; the orders of
  these groups and the orders of their automorphism groups are described in
  various places in said book. Simply for the convenience of the reader, we
  refer to {\tt en.wikipedia.org/wiki/List\_of\_finite\_simple\_groups\#Summary}
  for a concise list of these orders.} of simple groups and their outer
  automorphism groups. Note that ${\rm PKer}(G)/O_\infty(G)$ embeds into
  $\Aut(T_1)\times\ldots\times\Aut(T_\ell)$. Assume, for a contradiction, that
  $p$ divides $|G:O_{\infty}(G)|$. By assumption, $p\nmid |\Aut(T_i)|$ for each
  $i$, which forces $p\mid |G:{\rm PKer}(G)|$ and  $p\leq \ell$. Every $T_i$ has
  even order by the Odd-Order Theorem, see \cite[p.\ 2]{wilson}, so $2^\ell\mid
  n$ and $\nu_2(n)\geq \ell$; now $p\leq \ell$ contradicts $p>\nu_2(n)$, which
  we have shown above. This forces $p\nmid |G:O_\infty(G)|$, so  the Sylow
  $p$-subgroup $P$ lies in $O_\infty$. Since $p$ is also isolated in
  $|O_{\infty}(G)|$, the proof of the solvable case shows $P\unlhd
  O_{\infty}(G)$. Since $O_\infty(G)$ is characteristic in $G$, we know that $P$
  is  normal in~$G$.

  In conclusion, for every $p\in \bppi(n)$ there is a normal Sylow
  $p$-subgroup $G_p$ in $G$. Since $n\in\ourset$ and $p>\log\log n$, this
  subgroup is cyclic of size $p$. If  $p,q\in\bppi(n)$ are distinct, then
  $G_p,G_q\unlhd G$ implies that $G_qG_p\cong G_p\times G_q\cong C_{pq}$ is cyclic
  of order $pq$. Thus, $G$ has a normal cyclic Hall
  $\bppi(n)$-subgroup $B$, and  $G=H\ltimes B$ for some $H\leq G$ by the
  Schur-Zassenhaus Theorem \cite[(9.1.2)]{rob}.

\medskip
  
It remains to prove that $|H|\leq  (\log n)^{O((\log\log n)^c)}$ for some $c$.
Recall from Definition \ref{def:ups} that $n=ab$ such that $a\leq (\log n)^{\log
\log n}$ and if $p\in \bppi(n)$, then $p\mid b$ and $p^2\nmid n$; moreover, $b$
is squarefree and has at most $2\log \log n$ prime divisors. We have $|H|=ab/z$,
where $z$ is the product of the primes in $\bppi(n)$. It remains to show that
$b/z \leq  (\log n)^{O((\log\log n)^c)}$. If $p$ is a prime divisor of $b/z$,
then $p>\log\log n$ and $p\notin\bppi(n)$, that is, $p^2\nmid n$ and $p\mid n$
is isolated, but not strongly isolated. By definition, this means that there is
some non-abelian simple group $T$ of order dividing $n$ with $p\mid |T|$. We
show below that each such $T$ satisfies $|T|\leq (\log n)^{O((\log \log n))}$,
in particular, $p\leq (\log n)^{O((\log \log n))}$. Since $b/z$ has at most
$2\log \log n$ different prime divisors, we deduce $b/z \leq (\log n)^{O((\log
\log n)^2)}$; this then completes the proof of the
theorem.

 One can see from the known factorized orders of the finite non-alternating non-abelian
 simple groups that every such group  $T$ has a \emph{distinguished} prime power
 divisor $r^m \mid |T|$ with $m>1$ and $|T|\leq (r^m)^{O(m)}$: This is trivially
 true for the 26 sporadic groups; for the other non-alternating simple groups
 this follows because they are representable as quotients of groups of $d\times
 d$ matrices over a field of order $r^e$, and then $m=de$. Now if the order of
 $T$ divides $n\in\ourset$, then $m>1$ forces $r\leq \log\log n$ and $r^m\leq
 \log n$, hence $m\leq \log\log n$, and $|T|\leq (\log n)^{O((\log \log n))}$,
 as claimed.
 
 If  $T\cong \mathrm{Alt}_k$  is alternating of order $k!/2\leq k^k=2^{k\log  k}$, 
 then the  distinguished prime power divisor is $2^{\nu_2(k!)-1}$.
 Legendre's formula \cite[Theorem 2.6.4]{moll} shows  that $\nu_2(k!)=k-s_2(k)$
 where $s_2(k)\leq \log k$ is the number of $1$'s in the $2$-adic representation
 of $k$. Since $n\in\ourset$, we have $2^{\nu_2(k!)-1}\leq \log n$, so
 $k-\log(k)-1\leq \nu_2(k!)-1\leq \log\log n$. This shows that $2^k\leq 2k\log
 n$, and so $|T|\leq (2k\log n)^{\log k}$. Note that $|T|=k!/2$ divides $n$, and
 so  Stirling's formula $\ln (k!)=k\ln k-k+O(\ln k)$ shows that $k\leq \log 2n$
 for large enough $k$. This yields $|T|\leq (\log n)^{O(\log\log n)}$, as
 claimed.
\end{proof}

\begin{proof}[Proof of Theorem~\ref{thm:iso}]
  Let $G$ be a group of  order $n\in\ourset$. By Theorem \ref{thm:split} and
  Corollary~\ref{corB}, we can construct generators and a membership test for the
  cyclic Hall $\bppi(n)$-subgroup $B$ of $G$. With this we can use Proposition
  \ref{prop:pres} to construct a complement $H$, thus $G=H\ltimes B$. The
  complexity statement follows from the results we have used.
\end{proof}


{
\bibliographystyle{alpha}

\begin{thebibliography}{10}
  
  

\bibitem{Babai:quasi}
  L.\ Babai. Graph isomorphism in quasi-polynomial time. arXiv:1512.03547.

\bibitem{babaibeals}
L.\ Babai, R.\ Beals. A polynomial-time theory of black box groups. I.  Groups St.\ Andrews 1997 in Bath, I, 30–64,
London Math. Soc. Lecture Note Ser., 260, Cambridge Univ. Press, Cambridge, 1999. 
  
\bibitem{BCQ}
 L.\ Babai,  P.\ Codenotti, Y.\ Qiao.
 Polynomial-time isomorphism test for groups with no abelian normal subgroups. 
 LNCS - Automata, Languages, and Programming - Proceedings, ICALP, Springer, Warwick, UK, (2012) 51-62.

\bibitem{BQ}
  L.\ Babai, Y.\ Qiao.
  Polynomial-time isomorphism test for groups with abelian Sylow towers. Stacs, 
  Dagstuhl Publishing, Paris, (2012) 453-465.

\bibitem{BNV:enum}
S.\ R.\ Blackburn, P.\ M.\ Neumann, G.\ Venkatarman. 
Enumeration of finite groups. Vol 173, Cambridge Tracts in Math. Cambridge Univ.\ Press, 2007.

\bibitem{feit}
 N.\  Bourbaki. Algebra I. Ch.\ 1-3. Translated from the French. Reprint of the 1989 English translation. Elements of Mathematics. Springer, Berlin, 1998.

\bibitem{das}
B.\ Das, S.\ Sharma. Nearly Linear Time Isomorphism Algorithms for Some Nonabelian Group Classes.
Theory Comput.\ Sys.\  (in press, 2020).
 
\bibitem{cgroups}
H.\ Dietrich, D.\ Low.
Generation of finite groups with cyclic Sylow subgroups. J.\ Group Theory 24 (2021) 161--175.

\bibitem{PART2}
 H.\ Dietrich, J.\ B.\ Wilson. Isomorphism testing of groups of cube-free order. J.\ Algebra 545 (2020) 174--197.

\bibitem{glassbox}
 H.\ Dietrich, J.\ B.\ Wilson. Glass-box Algebra: Algebraic computation with verifiable
axioms. (in preparation) 

  
\bibitem{ErdosPalfy}
  P.\ Erd\H{o}s, P.\ P\'alfy.
  On the orders of directly indecomposable groups.  Discrete Math.\ 200 (1999) 165--179.

\bibitem{latin}
 A.\ B.\ Evans. Orthogonal Latin squares based on groups. Developments in Mathematics 57, Springer, 2018.
 

\bibitem{promise}
O.\ Goldreich. On promise problems: a survey. Theoretical computer science, 254--290,
Lecture Notes in Comput.\ Sci.\ 3895, Springer, Berlin, 2006.

 \bibitem{GQ}
  J.\ A.\ Grochow, Y.\ Qiao.
  Algorithms for group isomorphism via group extensions and cohomology. SIAM J.\ Comput., 46 (2017) 1153--1216.


  
\bibitem{Ili}
  C.\ S.\ Iliopoulos. Computing in general abelian groups is hard. Theoret.\ Comput.\ Sci.\ 41 (1985) 81--93.

\bibitem{KLM}
  W. M.\ Kantor, E.\ M.\ Luks, P.\ D.\ Mark.
  Sylow subgroups in parallel. J.\ Algorithms 31 (1999) 132--195.
     
\bibitem{KP}
  G.\ Karagiorgos and D.\ Poulakis. An algorithm for computing a basis of a finite abelian group. Algebraic informatics, 174--184, Lecture Notes in Comput.\ Sci., 6742, Springer, Heidelberg, 2011.
  
\bibitem{kavitha}
 T.\ Kavitha. Linear time algorithms for abelian group isomorphism and related problems. J.\ Comput.\ System Sci.\ 73 (2007) 986--996.


 
\bibitem{LeGall}
  F.\ Le Gall.
  Efficient isomorphism testing for a class of group extensions.  STACS (2009) 625--636.

\bibitem{Li-Qiao}
Y.\ Li, Y.\ Qiao. Linear algebraic analogues of the graph isomorphism problem and the Erd{\H o}s-R\'enyi model.  2017 IEEE 58th FOCS, 463--474.

\bibitem{Luks:northeastern}
E.\ M.\ Luks. Lectures on polynomial-time computation in groups. Technical Report NU-CCS-90-16, Northeastern University, 1990.


\bibitem{miller78}
  G.\ L.\  Miller On the $n^{\log n}$ isomorphism technique: A preliminary report. In: Proceedings of the Tenth Annual ACM symposium on Theory of computing, pp.\ 51--58. ACM (1978).
 
  
\bibitem{miller79}
  G.\ L.\ Miller.
  Graph isomorphism, general remarks. J.\ Comp.\ Sys.\ Sci.\ 18 (2) (1979) 128 - 142.

\bibitem{moll}
  V.\ H.\ Moll.
  Numbers and functions. Student Mathematical Library, 65. AMS, Providence, RI, 2012.
 
\bibitem{Papa}
C.\ H.\ Papadimitriou. Computational Complexity Addison-Wesley, 1994.

\bibitem{pet}
H.\ Peterson. Sorting and Element Distinctness on One-Way Turing Machines. In Language and Automata Theory and Applications, 2008, Springer, 433-439.


\bibitem{hardy}
R.\ G.\ Pinsky. Problems from the discrete to the continuous. Probability, number theory, graph theory, and combinatorics. Universitext. Springer, Cham, 2014.

\bibitem{ids}
S.\ Rajagopalan, L.\ J.\ Schulman. Verification of identities. SIAM J.\ Comput. 29 (2000) 1155-1163.

\bibitem{rob}
D.\ J.\ S.\ Robinson.
A Course in the Theory of Groups.
Springer, 1982. 

\bibitem{rosenbaum}
D.\ J.\ Rosenbaum, F.\ Wagner. Beating the generator-enumeration bound for $p$-group isomorphism. Theoret.\ Comput.\ Sci.\ 593 (2015) 16--25.


\bibitem{schonhage}
A.\ Sch\"onhage,  A.\ F.\ W.\ Grotefeld, E.\ Vetter.
Fast algorithms. 
A multitape Turing machine implementation. Bibliographisches Institut, Mannheim, 1994.

 \bibitem{Seress}
   \'A.\ Seress. Permutation group algorithms. Cambridge University Press 152, Cambridge, 2003.

    
\bibitem{simsbook} C.\ C.\ Sims. Computation with finitely presented groups. Cambridge University Press, 1994.

\bibitem{hermite}
  A.\ Storjohann, G.\ Labahn. Asymptotically fast computation of hermite normal forms of integer matrices.
  ACM Proc.\ 1996 Int.\ Symp.\ Symb.\ Alge.\ Comp (ISSAC) (1996) 259--266.
  
\bibitem{vikas}
  N.\  Vikas. An $O(n)$ algorithm for Abelian $p$-group isomorphism and an $O(n\log n)$ algorithm for Abelian group isomorphism, J.\ Comput.\ System Sci.\ 53 (1996) 1--9. 
   
\bibitem{Wilson:profile}
J.\ B.\ Wilson.  
The threshold for subgroup profiles to agree is logarithmic.
Theory of Comp.\ 15 (2019) 1--25.

\bibitem{wilson}
  R.\ A.\ Wilson. The finite simple groups. Graduate Texts in Mathematics, 251. Springer, London, 2009.
  
 \end{thebibliography}

 }

\end{document}